\documentclass[12pt]{article}
\usepackage{cite}      
\usepackage{graphicx}  
\usepackage{indentfirst}
\usepackage{amsmath}
\usepackage{amssymb}
\usepackage{amsthm}
\usepackage{amsfonts}
\usepackage{setspace}

\usepackage{anysize}
\usepackage{fancyhdr}
\title{Relay-Assisted User Scheduling in Wireless Networks with Hybrid-ARQ}
\author{Caleb K. Lo$^{1}$, John J. Hasenbein$^{2}$, Sriram Vishwanath$^{1}$ and Robert W. Heath, Jr.$^{1}$ \\ 
\\
$^{1}$Wireless Networking and Communications Group \\ 
Department of Electrical and Computer Engineering \\ 
The University of Texas at Austin \\ 
1 University Station C0803 \\ 
Austin, TX 78712-0240 \\ 
Phone: (512) 471-1190 \\ 
Fax: (512) 471-6512 \\ 
Email: \{clo, sriram, rheath\}@ece.utexas.edu \\ 
$^{2}$Graduate Program in Operations Research and Industrial Engineering \\ 
Department of Mechanical Engineering \\ 
The University of Texas at Austin \\ 
Austin, TX 78712-1063 \\ 
E-mail: jhas@mail.utexas.edu}
\date{}

\begin{document}

\singlespacing

\maketitle

\begin{abstract}
{This paper studies the problem of relay-assisted user scheduling for downlink wireless transmission.  The base station or access point employs hybrid automatic-repeat-request (HARQ) with the assistance of a set of fixed relays to serve a set of mobile users.  By minimizing a cost function of the queue lengths at the base station and the number of retransmissions of the head-of-line packet for each user, the base station can schedule an appropriate user in each time slot and an appropriate transmitter to serve it.  It is shown that a priority-index policy is optimal for a linear cost function with packets arriving according to a Poisson process and for an increasing convex cost function where packets must be drained from the queues at the base station.}
\end{abstract}

\doublespacing

Keywords - Relays, scheduling policies, hybrid automatic-repeat-request, priority-index rules.

\section{Introduction}\label{intro}

Relay-assisted communication will increase system throughput and coverage in local and metropolitan area networks \cite{RelaTaskGrou}.  The focus of this paper is on quantifying the scheduling-related benefits of relay-assisted communication.  For downlink transmission, this raises the question of how the presence of fixed relays impacts the user scheduling decisions at the base station or access point.  Relay-assisted scheduling has been studied in a variety of networks\footnote{Note that as the number of relays increases, communication between the base station and the relays becomes more challenging.  In particular, the level of signaling overhead increases, and intelligent frequency reuse planning for multi-relay transmission becomes more difficult.} \cite{VisMuk:PerfCellNetwRela:Sep:05, TasEph:StabPropConsQueu:92, NeeETAL:DynaPoweAlloRout:Jan:05, YehBer:ThroOptiContCoop:Oct:07}.

The benefits of relaying are enhanced by intelligent reception strategies such as HARQ \cite{ZhaVal:PracRelaNetwGene:Jan:05}.  Most of the prior work on HARQ-based scheduling does not account for the potential benefits of assisting relays, though.  Relays can be employed to decrease the number of HARQ transmissions that are required to serve a particular user.  Given a HARQ transmission framework, the relay-assisted scheduling problem is not merely a function of user queue lengths at the base station \cite{VisMuk:PerfCellNetwRela:Sep:05}.  It is now important to also consider the number of HARQ transmissions that have occurred for each user's head-of-line (HoL) packet, which directly impacts the decoding delay that each user incurs.  Decoding delay also depends on queue lengths \cite{BerGal:DataNetw:91} and has influenced work in the scheduling domain \cite{LuYao:OptiContFluiSide:Oct:03}.

In this paper we derive cost-minimizing user scheduling policies for a relay-based modification of the HARQ model in \cite{HuaETAL:WireScheHybrARQ:Nov:05}.  In \cite{HuaETAL:WireScheHybrARQ:Nov:05}, packets arrive at the base station, and in each time slot one of the users is scheduled.  Each user has a cost function that depends on its queue length at the base station and the number of retransmissions that have occurred for its HoL packet.  The objective is to schedule a user to minimize the long-term average expected cost.  It is shown in \cite{HuaETAL:WireScheHybrARQ:Nov:05} that the optimal scheduler is a fixed priority-index policy \cite{Kli:TimeSharServSyst:74}, where an index, or ``priority,'' is calculated for each user.  The users are ranked according to their priority indices, and the highest-ranked user with a nonempty queue is serviced in that time slot.  One example of a priority index is the ratio of the storage cost at the base station for a given user's packet to the expected time before that user decodes that packet \cite{HuaETAL:WireScheHybrARQ:Nov:05}.

Since the analysis in \cite{HuaETAL:WireScheHybrARQ:Nov:05} does not consider the presence of relays, it is unclear as to whether a priority-index policy is optimal for our relay-assisted system, as the priority index for each user can be adjusted if at least one relay has previously decoded its HoL packet.
To address this issue, we consider relay-assisted variants of the two problems in \cite{HuaETAL:WireScheHybrARQ:Nov:05}.  One problem entails minimizing a linear cost function with Poisson arrivals at the base station (LPA), while the other problem entails minimizing an increasing convex function of queue length without any new arrivals at the base station (DC).  We prove that the optimal scheduler for the relay-assisted variants of the LPA and DC problems is actually a priority-index policy as in \cite{HuaETAL:WireScheHybrARQ:Nov:05}.

\section{System Model}\label{system-model}
First, we introduce the notation used throughout the paper.  $|z|^2$ denotes the absolute square of a complex number $z$.  $\mathbb{E}$ denotes the expectation operator.

Consider the system in Fig. \ref{sys-model}, which consists of a single base station, $M$ fixed relays, and $N$ users.  Packets for each user arrive at the base station, and each packet is placed in a queue for its intended user.  The packet arrival processes are mutually independent.  Let $h_{i,n}$ denote the channel between nodes $i$ and $n$.  In each time slot, the base station selects a packet, which is then transmitted by the base station or a selected relay.  This relay must have previously decoded that packet.  If the scheduled user decodes the packet, the base station flushes it from its queue, and each relay removes the packet from its memory.  If the scheduled user cannot decode the packet, though, it remains in its queue at the base station.  Each relay also retains the packet in its memory, and the base station may either select that packet or a packet intended for another user for the next time slot.

When a packet is transmitted, its intended mobile target and all of the relays that have yet to decode it employ a generic HARQ decoding strategy.  For example, each of these receiving nodes can use maximal-ratio combining of successive transmissions of this packet, with the objective of improving its decoding probability.


We assume that before each time slot, the base station knows its channel gain $|h_{t,i}|^2$ to each user $i$ and the channel gain $|h_{a,i}|^2$ from each relay $a$ to each user $i$.  This is reasonable for a cellular network with a relatively small number of users $N$ and relays $M$.  As $N$ and/or $M$ grow large, though, the level of signaling overhead would become prohibitive for proper network operation.  We also assume that time is slotted, and each channel gain $|h_{i,n}|^2$ remains constant over a single HARQ retransmission sequence, which consists of a finite number of slots.  This assumption is reasonable in a slow fading environment.  Each channel gain $|h_{i,n}|^2$ also varies independently from one HARQ retransmission sequence to the next, which is a block fading assumption.

Given the transmission model, the relays are only considered in the scheduling policy when a selected user fails to decode its transmitted packet, requiring its future retransmission.  Each relay can store one packet for each user, where 1) this packet has been transmitted by the base station and 2) its intended mobile target failed to decode it.  Thus, the packet arrival process at each relay can be described as follows: assuming that there is at least one nonempty queue at the base station, one packet arrives in each time slot.  The packet arrival probability for a given user is the probability of that user being scheduled by the base station in that time slot.

Let $\textbf{S}(n)$ from \cite{HuaETAL:WireScheHybrARQ:Nov:05} be the state vector for the base station at time slot $n$.  Thus, $\textbf{S}(n)$ includes the number of transmission attempts for the current HoL packet for each user and the queue length for each user.  Also, let $\textbf{S}_a(n) = \{R_{a,1}(n),R_{a,2}(n),\ldots,R_{a,N}(n)\}$ be the state vector for relay $a$ at time slot $n$, where $R_{a,i}(n) = 1$ if relay $a$ has decoded the HoL packet for user $i$, but user $i$ has not decoded it.  Otherwise, $R_{a,i}(n) = 0$.  Let $\mathcal{M} = \{BS,1,2,\ldots,M\}$ denote the set of allowed transmitters.  Our objective is to design a scheduling policy $\pi_R\in\Pi_R$ such that $\pi_R(\textbf{S}(n),\textbf{S}_1(n),\textbf{S}_2(n),\cdots,\textbf{S}_M(n)) = (i,a)$, where transmitter $a\in\mathcal{M}$ serves the scheduled user $i$.

\section{Relay-Assisted Linear Poisson Arrivals Problem}\label{ralpa-main}
We consider the relay-assisted linear Poisson arrivals (RLPA) problem, which is a variant of the LPA problem in \cite{HuaETAL:WireScheHybrARQ:Nov:05}.  In the LPA problem, packets arrive at their corresponding queues at the base station.  The arrival process of the packets for user $i$ is Poisson with rate $\lambda_i$.  Let $c_{i,r_i}$ denote the cost of storing a packet for user $i$ that has already been transmitted $r_i$ times.  The base station computes a cost function $U_i$ for user $i$ that both depends on $r_i$ and is linear in the queue length $x_i(n)$, where
\begin{equation}
U_i(x_i(n),r_i^{HoL}(n)) = \left\{ \begin{array}{ll}
c_{i,0}(x_i(n)-1)+c_{i,r_i^{HoL}(n)} & x_i(n) > 0 \\
0 & x_i(n) = 0
\end{array} \right.
\end{equation}
\begin{equation}
0\leq c_{i,r_i}\leq c_{i,r_i^{'}},\quad r_i < r_i^{'} \nonumber
\end{equation}
implying that the storage cost is a nondecreasing function of the number of transmission attempts.


The RLPA problem entails determining the scheduling policy $\pi_R\in\Pi_R$ that minimizes the long-run average expected cost
\begin{eqnarray}
J_{RLPA} & = & \lim_{\tau\rightarrow\infty}\frac{1}{\tau}\mathbb{E}_{\pi_R}\Bigg[\sum_{n=1}^{\tau}\sum_{i=1}^{N}U_i(x_i(n),r_i^{HoL}(n))\Bigg]. \nonumber
\end{eqnarray}

The optimal policy for the RLPA problem is based on the fixed priority-index policy that is optimal for the LPA problem \cite{HuaETAL:WireScheHybrARQ:Nov:05}.  For the RLPA problem, knowledge of $\{\textbf{S}_1(n),\textbf{S}_2(n),\ldots,\textbf{S}_M(n)\}$ at the base station is useful in deciding which users can be served more quickly than others.  Note that cost increases with the number of retransmission attempts and the incurred delay.

\newtheorem{rlpa}{Theorem}
\begin{rlpa}\label{rlpa}
The optimal scheduling policy for the RLPA problem is a priority-index rule, where the HoL packet with the highest priority index over all nonempty base station queues is selected.  The transmitter that yields the highest priority index transmits the selected HoL packet.
\end{rlpa}

\begin{proof}
The proof is in Appendix \ref{rlpa-proof}.
\end{proof}

We now provide an intuitive justification of Theorem \ref{rlpa}.  In \cite{HuaETAL:WireScheHybrARQ:Nov:05}, each user $i$ is assigned a fixed priority index $c_{i,r_i^{HoL}}/T_{i,r_i^{HoL}}$, where $c_{i,r_i^{HoL}}$ is the holding-cost rate for the HoL packet of user $i$ that has undergone $r_i^{HoL}$ transmission attempts and $0\leq r_i^{HoL}\leq r_i^{max}$.  Also, $T_{i,r_i^{HoL}}$ is the expected service time for the HoL packet of user $i$, where
\begin{equation}
T_{i,r_i^{HoL}} = 1 + \sum_{b=r_i^{HoL}}^{r_i^{max}-1}\prod_{l=r_i^{HoL}}^{b}g_i(l) \label{exp-serv-time}
\end{equation}
and $g_i(l)$ is the probability of a decoding failure by user $i$ given that its HoL packet has been transmitted $l$ times.  The optimal policy from \cite{HuaETAL:WireScheHybrARQ:Nov:05} is to schedule the user with the highest priority index.  To simplify the following discussion, we have not considered the impact of relaying in \eqref{exp-serv-time}.

A lower value of $T_{i,r_i^{HoL}}$ implies that user $i$ achieves a higher priority index.  Now consider the two-user system in \cite[Figure 4]{HuaETAL:WireScheHybrARQ:Nov:05}, which we have re-plotted as Fig. \ref{hua05-fig4}.  Here, a new arrival for the first user has priority over a retransmission for the second user.  No relays are present, which is equivalent to the RLPA problem if $\textbf{S}_a(n) = (0,0)\quad\forall a\in\{1,2,\ldots,M\}$.  Assume that only relay $a$ has decoded the HoL packet for user 2, and relay $a$ can decrease $T_{2,r_2^{HoL}}$ since $|h_{a,2}|^2 > |h_{t,2}|^2$.  This increases user 2's priority, and if $|h_{a,2}|^2$ is above a threshold value, user 2 can obtain a higher priority index than user 1, and so a retransmission for user 2 has priority over a new arrival for user 1.  In particular, $\pi_R(\textbf{S}(n),\textbf{S}_1(n),\textbf{S}_2(n),\ldots,\textbf{S}_M(n)) = (2,a)$ as opposed to $\pi(\textbf{S}(n)) = 1$ for the LPA problem.

Thus, the introduction of relays for the RLPA problem results in a modification to the optimal policy in \cite{HuaETAL:WireScheHybrARQ:Nov:05}.  Users are still sorted according to their priority indices and the highest priority user with a nonempty queue is scheduled.  In this case, though, each relay $a$ can inform the base station of its ability to improve the priority indices of some subset of the users by reporting $\textbf{S}_a(n)$ to the base.  The base station can calculate an improved priority index for each user $i$ such that $R_{a,i}(n) = 1$.  Then, all priority indices including any revised indices are sorted, and the highest priority user with a nonempty queue is scheduled along with the transmitter that yields that highest priority index.

\section{Relay-Assisted Draining Convex Problem}
Now we consider the relay-assisted draining convex (RDC) problem, which is a variant of the DC problem in \cite{HuaETAL:WireScheHybrARQ:Nov:05}.
The DC problem is a draining problem where no new packets arrive at the base station, and the base station wants to empty all of the user queues.  The base station computes a cost function $U_i$ for user $i$, where $U_i$ is an arbitrary increasing function of the queue length $x_i(n)$ and is independent of the number of transmission attempts of the HoL packet of user $i$.  Thus,
\begin{eqnarray}
U_i(x_i(n),r_i^{HoL}(n)) & = & U_i(x_i(n)). \nonumber
\end{eqnarray}

The base station initially has a set of packets $(x_1(1),x_2(1),\ldots,x_N(1))$.  The RDC problem entails determining the scheduling policy $\pi_R\in\Pi_R$ that minimizes the total expected draining cost
\begin{eqnarray}
J_{RDC} & = & \mathbb{E}_{\pi_R}\Bigg[\sum_{n=1}^{\infty}\sum_{i=1}^{N}U_i(x_i(n))\Bigg]. \nonumber
\end{eqnarray}

As in Section \ref{ralpa-main}, the optimal policy for the RDC problem is based on the fixed priority-index policy that is optimal for the DC problem \cite{HuaETAL:WireScheHybrARQ:Nov:05}.  For the RDC problem, knowledge of $\{\textbf{S}_1(n),\textbf{S}_2(n),\ldots,\textbf{S}_M(n)\}$ at the base station is useful in deciding which users can be served more quickly than others.  Note that as in the RLPA problem, cost increases with the incurred delay.

\newtheorem{rdc}[rlpa]{Theorem}
\begin{rdc}\label{rdc}
The optimal scheduling policy for the RDC problem is a priority-index rule, where the HoL packet with the highest priority index over all nonempty base station queues is selected.  The transmitter that yields the highest priority index transmits the selected HoL packet.
\end{rdc}

\begin{proof}
The proof is similar to that in Appendix \ref{rlpa-proof}, so we provide a brief sketch of it as follows.  As in Appendix \ref{rlpa-proof}, we transform the RDC problem into an instance of Klimov's multiclass queueing problem \cite{Kli:TimeSharServSyst:74}.  This implies that the transformed problem, which we refer to as the RDCK problem, has an optimal priority index policy.  Finally, it can be shown that this policy is also optimal for the RDC problem.

Based on Appendix \ref{rlpa-proof}, in the RDCK problem each user $i$ has $K_i = (M+1)x_i(1)(r_i^{max}+1)$ queues, and each queue is labeled as $(i,r_i,x_i,l)$.  A packet in $(i,r_i,x_i,l)$ has been transmitted $r_i$ times, has been decoded by relay $d_{i,l}$ and has not been decoded by relay $d_{i,m}$ for $l < m\leq M$.

Now, the objective is to find $\pi_R\in\Pi_R$ that minimizes
\begin{eqnarray}
J_{RDCK} & = & \mathbb{E}_{\pi_R}\Bigg[\sum_{n=1}^{\infty}\sum_{(i,r_i,x_i,l)\in\Omega}\textbf{1}_{i,r_i,x_i,l}(n)U_i(x_i)\Bigg] \nonumber
\end{eqnarray}
where
\begin{equation}
\textbf{1}_{i,r_i,x_i,l}(n) = \left\{ \begin{array}{ll}
1 & (i,r_i,x_i,l) \textnormal{ is nonempty in slot n} \\
0 & \textnormal{otherwise}. \nonumber
\end{array} \right.
\end{equation}

It follows from \cite[Theorem 2]{HuaETAL:WireScheHybrARQ:Nov:05} and \cite[Lemma 2]{HuaETAL:WireScheHybrARQ:Nov:05} that the optimal policy for the RDCK problem assigns queue $(i,r_i^{'},x_i,m)$ higher priority than queue $(i,r_i,x_i,l)$ for all $i$, $x_i$, $r_i^{'} > r_i$ and $m\geq l$.

Since the RDCK problem is a special case of Klimov's problem, the optimal policy for the RDCK problem is a priority-index rule.  We then transform the RDCK problem back to the RDC problem to conclude that the optimal policy for the RDC problem is also a priority-index rule.  The HoL packet with the highest priority index along with the transmitter that yields that index are selected over all nonempty queues at the base station.  Note that unlike the RLPA problem, the priority indices for the RDC problem do not admit closed-form expressions.
\end{proof}

The intuitive justification of Theorem \ref{rdc} is similar to that of Theorem \ref{rlpa}.  It should be noted that for the RLPA and RDC problems, a relay $a$ with $R_{a,i}(n) = 1$ for some user $i$ only increases the priority index of user $i$ if $|h_{t,i}|^2 < |h_{a,i}|^2$.

\section{Simulation Results}
Now we evaluate the performance of relaying in the RLPA problem.  Fig. \ref{relay-location} displays the impact of employing $M = 1$ relay in a system with $N = 2$ users and arrival rates $\lambda_1 = \lambda_2 = 0.3$.  The maximum number of retransmissions is $r_1^{max} = r_2^{max} = 2$, and the cost rates are $c_1 = [0.98~1~1.02]$ and $c_2 = [1.25~1.5~1.75]$.  We model the effects of limited transmit-side channel knowledge by assuming that each channel $h_{t,i}$ and $h_{1,i}$ undergoes Rayleigh fading and varies independently between time slots.  The base station only knows $\mathbb{E}(|h_{t,i}|^2)$ and $\mathbb{E}(|h_{1,i}|^2)$, while the relay only knows $\mathbb{E}(|h_{1,i}|^2)$.  Assuming the presence of only channel distribution information at each transmitter implies that the base station only knows its average channel parameters to users 1 and 2 as $\bar{\eta}_1 = \bar{\eta}_2 = 0.9$, and the relay only knows its average channel parameters to users 1 and 2 as $\bar{\eta}_{1,1}$ and $\bar{\eta}_{1,2}$, respectively.  The base station computes the user probability of decoding failure, assuming that no relays assist it, as
\begin{equation}\label{dec-prob}
g_i(r_i) = \left\{ \begin{array}{ll}
\bar{\eta}_i\cdot 0.9^{r_i} & 0\leq r_i < r_i^{max} \\
0 & r_i = r_i^{max}.
\end{array} \right.
\end{equation}
To characterize the performance impact of relaying, we vary $\bar{\eta}_{1,2}$ and fix $\bar{\eta}_{1,1} = 0.9$, as this limits the ability of the relay to assist user 1 and allows us to focus on how the relay assists user 2.

We run our simulation over $B$ time slots, where at most one packet can be decoded in each time slot.  If $D$ packets are decoded during this simulation run, we define the throughput as $D/B$.  We see that the long-term cost decreases and the throughput increases as the average channel gains from the relay to the base station and user 2 increase.  In particular, the cost decreases by $17.3\%$ and the throughput increases by $43.2\%$ when $\bar{\eta}_{1,2}$ increases from 0.1 to 0.9.  This demonstrates the performance gains via intelligent relay deployment.  Also, we see that the effects of limited channel knowledge at the base station and the relay are asymptotically negligible.



Fig. \ref{user2eta} shows how the optimal policy behaves as a function of the average base station channel gains to the users.  We adopt the same parameters as in Fig. \ref{relay-location}, except for $\bar{\eta}_{1,1} = \bar{\eta}_{1,2} = 0.15$, and we vary $\bar{\eta}_2$.

We see that the throughput of the optimal policy deteriorates as the average base station channel gain to user 2 decreases.  We also consider a case where no relay is present, and it can be seen that the throughput of the optimal policy decreases at an even faster rate than the case where $M = 1$ relay assists the base station.  This example further highlights the inherent challenges in a cellular network of servicing cell-edge users.

\section{Conclusion}
We have considered the problem of user scheduling in a downlink wireless system with HARQ retransmissions.  By allowing fixed relays to assist the base station or access point in servicing a scheduled user, a cost function of the user queue lengths at the base station and the number of retransmissions of the HoL packet for each user can be minimized.  We have proved that the optimal scheduler for the relay-assisted extensions of two problems in \cite{HuaETAL:WireScheHybrARQ:Nov:05} is a priority-index rule.

The main contribution of this work opens up several avenues for further investigation of the relay-assisted scheduling problem.  In particular, it is clear that relay-assisted scheduling is actually a relay selection problem, and extensive prior work on relay selection has conclusively shown its difficult cross-layer nature.  Thus, a more comprehensive approach to this problem would consider additional factors such as the specific type of HARQ being employed at the relays and the users along with more general packet arrival processes at the base station.  For example, if one user is downloading multimedia content while another is sending text messages, this could be used to design an appropriate cost function for each user.  Also, the performance impact of real-world issues such as timing mismatch between the base station and any of the relays in its network should be evaluated.  For example, if a delay of one time slot occurs before a selected relay receives its selection notification from the base station, then it will transmit and cause a packet collision with another selected relay or the base station, which could degrade the achieved throughput.

\appendix

\section{Proof of Theorem \ref{rlpa}}\label{rlpa-proof}
As in \cite{HuaETAL:WireScheHybrARQ:Nov:05}, we transform the RLPA problem into an instance of the multiclass queueing problem of Klimov \cite{Kli:TimeSharServSyst:74}.  Thus, the transformed problem, which we refer to as the RLPAK problem, has an optimal priority index policy.  We then show that this policy is also optimal for the RLPA problem.

Note that for each user $i$, the $M$ relays are sorted as $\{d_{i,1},d_{i,2},\ldots,d_{i,M}\}$, where $|h_{d_{i,1},i}|^2 < |h_{d_{i,2},i}|^2 < \cdots < |h_{d_{i,M},i}|^2$.  Then, each user $i$ has $(M+1)(r_i^{max}+1)$ queues, and each queue is labeled as $(i,r_i,l)$.  A packet in $(i,r_i,l)$ has been transmitted $r_i$ times, has been decoded by relay $d_{i,l}$, and has not been decoded by relay $d_{i,m}$ for $l < m\leq M$.  There are a total of $K = \sum_{i=1}^{N}(M+1)(r_i^{max}+1)$ queues.  If $\lambda = \sum_{i=1}^{N}\lambda_i$, each arriving packet is assigned to $(i,0,0)$ with probability $p_{i,0} = \lambda_i/\lambda$, and $(i,r_i,l)$ has a deterministic service time of $b_{i,r_i,l} = 1$ time slot.  The cost of storing a packet in queue $(i,r_i,l)$ is $c_{i,r_i,l}$ and the number of packets in queue $(i,r_i,l)$ at the beginning of the $n$th time slot is $x_{i,r_i,l}(n)$.  Fig. \ref{rlpak-figure} shows an example of the RLPAK problem for user 1 where $r_1^{max} = 2$.

The queue transition probabilities in the RLPAK problem are determined as follows.  Let $g_{i,l,k}(r_i)$ denote the probability that relay $d_{i,l}$ cannot decode the HoL packet of user $i$ after its transmission attempt $r_i$ by relay $d_{i,k}$.  Also, let $g_{i,l}(r_i,1)$ denote the probability that user $i$ cannot decode its HoL packet after its transmission attempt $r_i$ by relay $d_{i,l}$.  Then
\begin{equation}
\begin{array}{lll}
p_{(i,r_i,0),(i,r_i+1,0)} & = & g_i(r_i)g_{i,1,0}(r_i)g_{i,2,0}(r_i)\cdots g_{i,M,0}(r_i), \\
p_{(i,r_i,0),(i,r_i+1,l)} & = & g_i(r_i)(1-g_{i,l,0}(r_i))g_{i,l+1,0}(r_i)g_{i,l+2,0}(r_i)\cdots g_{i,M,0}(r_i) \\
p_{(i,r_i,l),(i,r_i+1,n)} & = & g_{i,l}(r_i,1)(1-g_{i,n,l}(r_i))g_{i,n+1,l}(r_i)g_{i,n+2,l}(r_i)\cdots g_{i,M,l}(r_i), n > l, \\
p_{(i,r_i,l),(i,r_i+1,l)} & = & g_{i,l}(r_i,1)g_{i,l+1,l}(r_i)g_{i,l+2,l}(r_i)\cdots g_{i,M,l}(r_i), \\
p_{(i,r_i,l),(i,r_i+1,n)} & = & 0, n < l \nonumber
\end{array}
\end{equation}
and so the packet departs the system from $(i,r_i,0)$, $(i,r_i,l)$ and $(i,r_i^{max},l)$ with probabilities $1-g_i(r_i)$, $1-g_{i,l}(r_i,1)$ and 1 respectively where $l\in\{0,1,\ldots,M\}$.

Thus, for any $A\subset\Omega = \{1,2,\ldots,K\}$ and any $(i,r_i,l)\in A$, the average total service time is
\begin{eqnarray}
T_{i,r_i,l}^{(A)} & = & 1 + \sum_{k,r_k,m}p_{(i,r_i,l),(k,r_k,m)}T_{k,r_k,m}^{(A)}. \nonumber
\end{eqnarray}

From the above discussion, it can be concluded that the RLPAK problem, which is a transformed version of the RLPA problem, is an instance of the multiclass queueing problem of \cite{Kli:TimeSharServSyst:74}.  This conclusion also relies on the simple queueing dynamics of the relay: 1) a packet only arrives at the relay if the base station has transmitted it, and 2) the relay automatically flushes a packet once it has been decoded by its intended user.  In addition, the base station automatically flushes a packet once it has been decoded by its intended user.  It should be noted that the state space of the RLPAK problem is an expanded version of that in the LPAK problem.

Now, the objective is to find $\pi_R\in\Pi_R$ that minimizes
\begin{eqnarray}
J_{RLPAK} & = & \lim_{\tau\rightarrow\infty}\frac{1}{\tau}\mathbb{E}_{\pi_R}\Bigg[\sum_{n=1}^{\tau}\sum_{(i,r_i,l)\in\Omega}c_{i,r_i,l}x_{i,r_i,l}(n)\Bigg]. \nonumber
\end{eqnarray}
To this end, we state the following result.

\newtheorem{rlpak-lemma}{Lemma}
\begin{rlpak-lemma}\label{rlpak-lemma}
Let $A_k$, $k=1,2,\ldots,K$ be the sets of queues generated by the Klimov algorithm in \cite[Section 3]{HuaETAL:WireScheHybrARQ:Nov:05} for the LPAK problem.  For each $k = 1,2,\ldots,K$ and for all $(i,r_i,l)\in A_k$: \\
1) $(i,r_i^{'},m)\in A_k$ \textnormal{for all} $r_i^{'} > r_i$ \textnormal{and for all} $m\geq l$. \\
2) $T_{i,r_i,0}^{(A_k)} = 1+\sum_{q=r_i}^{r_i^{max}-1}\prod_{s=r_i}^{q}g_i(s) = T_{i,r_i,0}^{(\Omega)}$. \\
3) $T_{i,r_i,m}^{(A_k)} = 1+\sum_{q=r_i}^{r_i^{max}-1}\prod_{s=r_i}^{q}g_{i,m}(s,1) = T_{i,r_i,m}^{(\Omega)}, m > 0$. \\
4) $T_{i,r_i^{'},m}^{(A_k)} \leq T_{i,r_i,l}^{(A_k)}$ \textnormal{for all} $r_i^{'} > r_i$ \textnormal{and for all} $m\geq l$. \\
5) $\alpha_k = \textnormal{arg}\min_{(i,r_i,l)\in A_k}(c_{i,r_i,l}/T_{i,r_i,l}^{(\Omega)})$.
\end{rlpak-lemma}

\begin{proof}
This result follows in a straightforward manner from \cite[Lemma 1]{HuaETAL:WireScheHybrARQ:Nov:05}.
\end{proof}

By combining \cite[Theorem 1]{HuaETAL:WireScheHybrARQ:Nov:05} and Lemma \ref{rlpak-lemma}, it follows that the optimal scheduling policy for the RLPAK problem is a priority-index rule where the priorities $\alpha_1,\alpha_2,\ldots,\alpha_K$ satisfy
\[\frac{c_{\alpha_1}}{T_{\alpha_1}^{(\Omega)}}\geq\frac{c_{\alpha_2}}{T_{\alpha_2}^{(\Omega)}}\geq\ldots\geq\frac{c_{\alpha_K}}{T_{\alpha_K}^{(\Omega)}}.\]

Since the optimal scheduling policy for the RLPAK problem is a priority-index rule, we employ \cite[Corollary 1]{HuaETAL:WireScheHybrARQ:Nov:05} to conclude the the optimal scheduling policy for the RLPA problem is also a priority-index rule.  The HoL packet with the highest priority index of $c_{i,r_i^{HoL}}/T_{i,r_i^{HoL}}$ along with the transmitter that yields that index are selected over all nonempty queues at the base station.

\begin{figure}[tb]
\begin{center}
\includegraphics[width=3.0in]{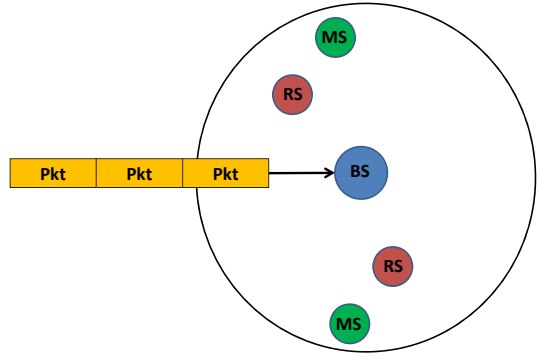}
\end{center}
\caption{Wireless network with relay-assisted scheduling.}
\label{sys-model}
\end{figure}

\begin{figure}[tb]
\begin{center}
\includegraphics[width=3.0in]{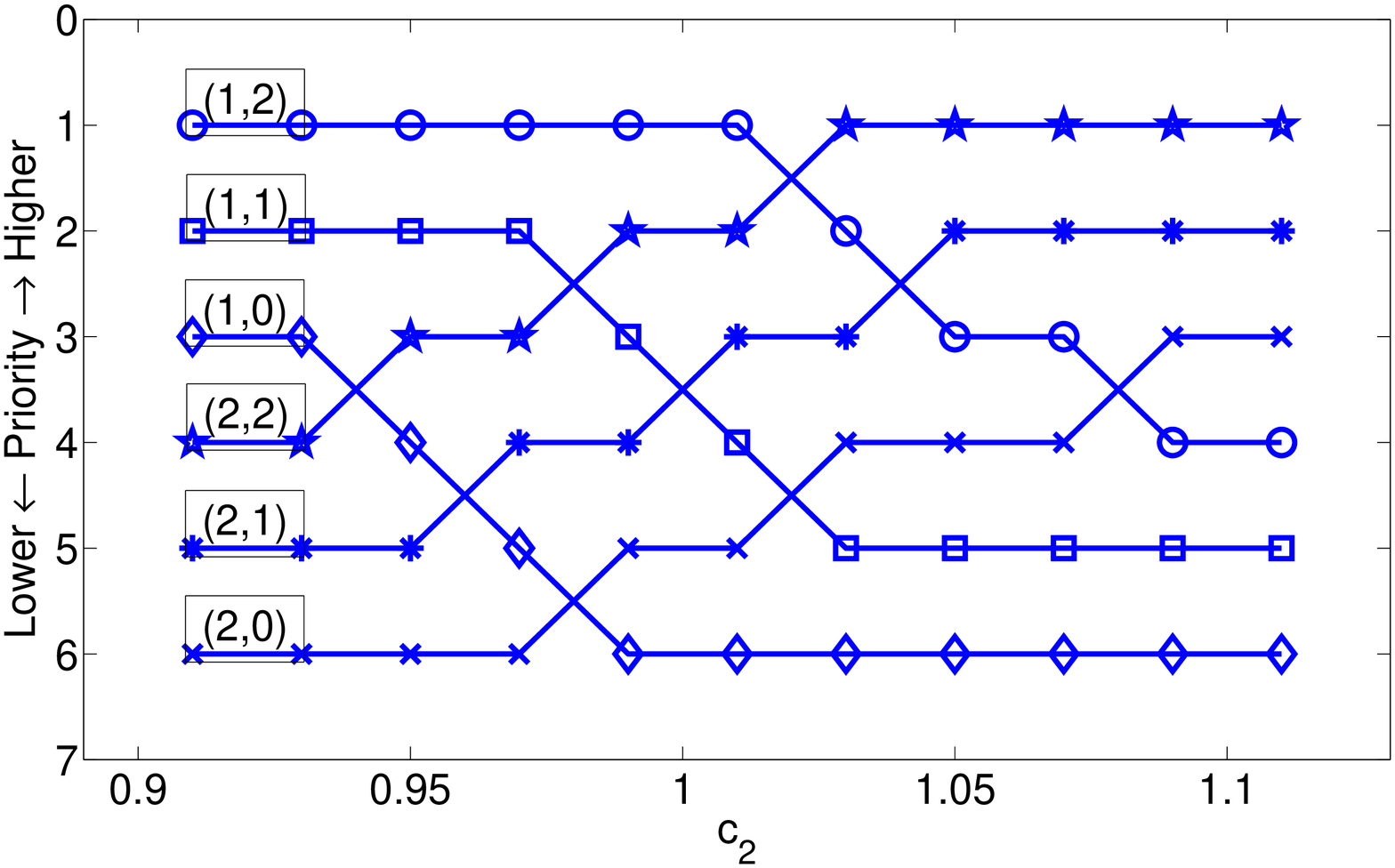}
\end{center}
\caption{Optimal priority orders versus holding-cost rate of user 2 in the LPAK problem.}
\label{hua05-fig4}
\end{figure}

\begin{figure}[tb]
\begin{center}
\includegraphics[width=3.0in]{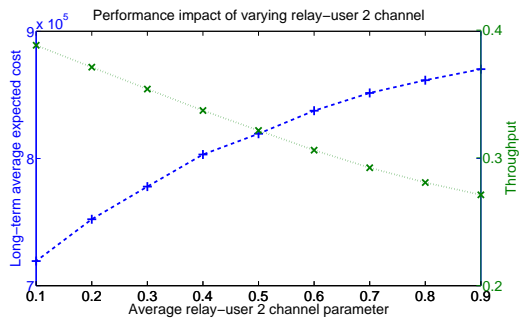}
\end{center}
\caption{Long-term average expected cost and throughput for RLPA problem as function of average channel from relay to user 2.}
\label{relay-location}
\end{figure}


\begin{figure}[tb]
\begin{center}
\includegraphics[width=3.0in]{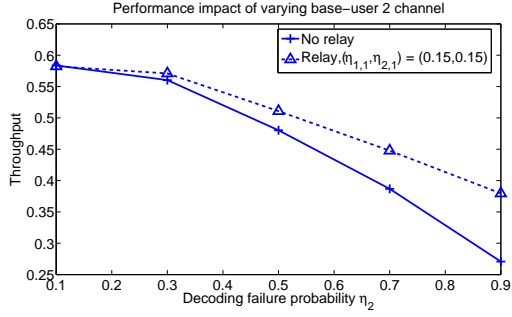}
\end{center}
\caption{Throughput for RLPA problem as function of average channel from base station to user 2.}
\label{user2eta}
\end{figure}



\begin{figure}[tb]
\begin{center}
\includegraphics[width=3.0in]{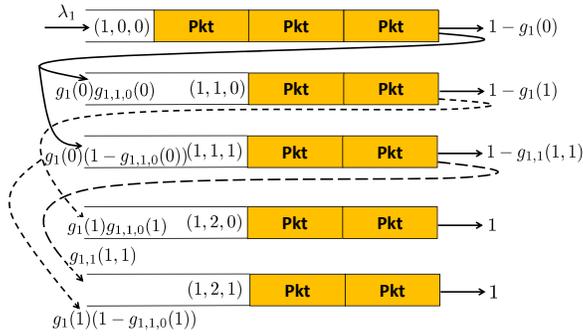}
\end{center}
\caption{System model for RLPAK problem with queues for user 1.}
\label{rlpak-figure}
\end{figure}

\end{document}